\documentclass{article}

\usepackage{arxiv}

\usepackage[utf8]{inputenc} 
\usepackage[T1]{fontenc}    
\usepackage{hyperref}       
\usepackage{url}            
\usepackage{booktabs}       
\usepackage{amsfonts} 
\usepackage{amsmath,amsthm} 
\usepackage{nicefrac}       
\usepackage{microtype}      
\usepackage{lipsum}		
\usepackage{graphicx}
\usepackage{natbib}
\usepackage{doi}

\newcommand{\reais}{I\!\! R}

\usepackage{hyperref}
\usepackage{mathtools} 
\usepackage{rotating}
\usepackage{IEEEtrantools}
\usepackage{comment}
\usepackage{algpseudocode, algorithmicx,algorithm}
\usepackage{adjustbox}
\usepackage{subfig}
\usepackage{tabularx,booktabs}
\usepackage{rotating}

\newtheorem{prop}{Proposition}
\newtheorem{lema}{Lemma}
\newtheorem{teorema}{Theorem}

\title{A note on solving the envy-free perfect matching problem with qualities of items}

\author{ Marcos Salvatierra\\ Normal Superior School \\ Amazonas State University 
}

\date{}

\hypersetup{
pdftitle={A template for the arxiv style},
pdfsubject={q-bio.NC, q-bio.QM},
pdfauthor={David S.~Hippocampus, Elias D.~Striatum},
pdfkeywords={First keyword, Second keyword, More},
}

\begin{document}
\maketitle

\begin{abstract}
In the envy-free perfect matching problem, $n$ items with unit supply are available to be sold to $n$ buyers with unit demand. The objective is to find allocation and prices such that both seller's revenue and buyers' surpluses are maximized -- given the buyers'  valuations for the items -- and all items must be sold. A previous work has shown that this problem can be solved in cubic time, using maximum weight perfect matchings to find optimal envy-free  allocations and shortest paths to find optimal envy-free prices. In this work, I consider that buyers have fixed budgets, the items have quality measures and so the valuations are defined by multiplying these two quantities. Under this approach, I prove that the valuation matrix have the inverse Monge property, thus simplifying the search for optimal envy-free allocations and, consequently, for optimal envy-free prices through a strategy based on dynamic programming. As result, I propose an algorithm that finds optimal solutions in quadratic time.

\end{abstract}






\section{Introduction}
\label{sec:intro}

Suppose a seller wants to sell $n$ similar items -- for example, cars -- and $n$ buyers are interested in buying one car each. The cars have different quality attributes such brand, numbers of doors, metallic paint, steering assistance, etc. Moreover, each buyer has a budget available to purchase a car. In this way, the seller will estimate the valuation that each buyer will have for each car, multiplying the buyer's budget by the quality measure of the car. Furthermore, each buyer wants to be happy with the price of the car sold to him/her, in the sense that the price is less than or equal to his/her valuation for that car, and furthermore, the car sold to him/her is the one that gives him/her the greatest surplus, i.e., the difference between the valuation for the car and the price of the car sold to him/her is greater than or equal to any other difference between the valuation and the price of another car.

So, formally, we have a set of $n$ items with unit supply, a set of $n$ buyers with unit demand, each buyer $i$ has a budget $v_i$, each item $j$ has a quality measure $q_j$, and then the valuation of buyer $i$ for item $j$ is given by $v_{i,j} = v_i q_j$. The valuation matrix is defined by $V = (v_{i,j}) \in \reais_{>0}^{n \times n}$. The seller needs to decide which item to sell to which buyer. In terms of decision variables, $x_{i,j}=1$ if buyer $i$ buys item $j$ and $x_{i,j}=0$ otherwise. In addition, the prices set should maximize both the seller's revenue and the buyers' surpluses, i.e., should maximize $\sum_{j=1}^n p_j$ and should satisfy $v_{i,j}x_{i,j} - p_j \geq v_{i,k}x_{i,j} - p_k$ for all $k \neq j$, if buyer $i$ buys item $j$, with $v_{i,j} x_{i,j} \geq p_j$. Finally, the seller needs to sell all items, i.e., $\sum_{j=1}^n x_{i,j} = 1$ for all $i = 1, \ldots , n$.

In this scenario, we say that the seller needs to find an optimal {\it envy-free allocation} for the buyers and optimal {\it envy-free prices} for the items. This is a special case of the {\it envy-free perfect matching problem}, which was studied by \cite{Arbib2019}. This last problem, in turn, is a particular case of the {\it unit-demand envy-free pricing problem}, proposed by \cite{Guruswami2005}. Other particular cases of the unit-demand envy-free pricing problem that can be solved in polynomial-time were also studied by \cite{gunluk2008pricing}, \cite{chen2011optimal} and  \cite{marcosms2021a14100279}, for example.

Having made these considerations, the main objective of this work is to propose an algorithmic strategy using a dynamic programming approach that finds optimal envy-free allocations and optimal envy-free prices in quadratic time for this particular case, thus reducing the time complexity proved by  \cite{Arbib2019} for the more general envy-free perfect matching problem.

\section{Mathematical properties}
\label{sec:math}

First, let's analyze the properties of the valuation matrix. Without loss of generality, for the remainder of this section we will assume that $v_i \geq v_j$ and $q_i \geq q_j$ if $i<j$. In this way, we have $v_{i,j} \geq v_{i,k}$ if $j < k$ and $v_{i,j} \geq v_{r,j}$ if $i < r$. So we have, for $i < j$,
 $$v_i (q_i - q_j) \geq v_j ( q_i - q_j ) \Rightarrow v_i q_i - v_i q_j \geq v_j q_i - v_j q_j \Rightarrow v_i q_i + v_j q_j \geq v_j q_i + v_i q_j.    $$

Therefore, $v_{i,i} + v_{j,j} \geq v_{j,i} + v_{i,j}$. This property is precisely the so called {\it inverse Monge property}. \cite{burkard1996perspectives} showed that such matrices has the maximum trace in the main diagonal, i.e., the identity matrix is the permutation that maximizes $tr(\Pi^TV)$, where $\Pi$ is a permutation matrix, and hence the optimal envy-free allocation for our problem is $x_{i,j} = 1$ if $i = j$ and $x_{i,j} = 0$, otherwise. From the point of view of graph theory, the identity matrix defines a maximum weight perfect matching in the bipartite graph whose costs are in matrix $V$, and from an economic perspective, this allocation maximizes social welfare, i.e., maximizes $\sum_{i=1}^n \sum_{j=1}^n x_{i,j} v_{i,j}$.

Once the optimal envy-free allocation is found, it remains to find the optimal envy-free prices. It can be clearly noted that the buyer with the lowest valuation for an item will buy this item for this price. I will formally prove this fact in the following proposition:

\begin{prop} \label{prop:minprice}
The optimal envy-free price of item $n$ is $p_n = v_{n,n}$.
\end{prop}

\begin{proof}
Indeed, if $p_n < v_{n,n}$, then the total revenue earned from selling the items will be decreased by $v_{n,n}-p_n$, instead of item $n$ being sold at the price $v_{n,n}$. On the other hand, if $p_n > v_{n,n}$, then the surplus of buyer $n$ will be negative, and the item will not be allocated to him.
\end{proof}

Now that we have the lowest optimal envy-free price set, a natural way to find the remaining optimal envy-free prices is to search upwards. Let's think intuitively: in a $2 \times 2$ valuation matrix, we already know that item 2 is allocated to buyer 2, and that the optimal envy-free price is $v_{2,2}$. Then, we must calculate a price for buyer 1 based on $v_{1,1}$ that maximizes his surplus. Thus, we must calculate surplus of buyer 1 if he were to buy item 2, that is, we must calculate $v_{1,2} - v_{2,2}$. Since $v_{1,2}-v_{2,2} \geq 0$, this value should be the maximum surplus of buyer 1, so the optimal envy-free price of item 1 will be $p_1 = v_{1,1}- ( v_{1,2} - v_{2,2} )$. Continuing with this reasoning, the next optimal envy-free prices will be calculated by observing the surpluses of the next buyers with the optimal envy-free prices already found.

This intuitive way of finding optimal envy-free prices will be generalized in the following lemma:

\begin{lema}\label{lemma:nextprice}
If the $k>1$ lowest optimal envy-free prices have already been found, the next optimal envy-free price $p_{n-k}$, in ascending order, is found as follows:
$$ p_{n-k} = v_{n-k, n-k} - \max_{i = (n - k+1), \ldots , n} \{ v_{n-k,i} - p_i  \}$$
\end{lema}

\begin{proof}
Indeed, the surpluses of buyer $n-k$ are:
\begin{IEEEeqnarray*}{rCl}
	v_{n-k, n-k+1}& - & p_{n-k+1}\\
	v_{n-k, n-k+2} & - & p_{n-k+2}\\
& \vdots & \\
v_{n-k, n} & - & p_n
\end{IEEEeqnarray*}

Thus, the maximum of these values must be subtracted from $v_{n-k,n-k}$ for the envy-free condition to hold. Furthermore, if any other value greater than this maximum is subtracted from $v_{n-k,n-k}$, the resulting price will decrease the seller's revenue.  
\end{proof}

\section{Algorithm for the problem}
\label{sec:algo}

The mathematical properties of the problem, explored in the previous section, allows me to present the algorithm below that finds an optimal envy-free allocation and optimal envy-free prices, given as input: $n$ (the number of items and buyers), $v_i$ for all $i = 1, \ldots, n$ (the buyers' valuations) and $q_j$ for all $j = 1, \ldots, n$ (the items qualities).

\begin{enumerate}
    \item Do a permutation $\pi: \{1, 2, \ldots, n  \} \rightarrow \{1, 2, \ldots, n  \}$ such that $v_{\pi^{-1}(1)} \geq v_{\pi^{-1}(2)} \geq \ldots \geq v_{\pi^{-1}(n)}$.
    
    \item Do a permutation $\sigma: \{1, 2, \ldots, n  \} \rightarrow \{1, 2, \ldots, n  \}$ such that $q_{\sigma^{-1}(1)} \geq q_{\sigma^{-1}(2)} \geq \ldots \geq q_{\sigma^{-1}(n)} $.
    
    \item Construct the valuation matrix $V = (v_{i,j})$ such that $v_{i,j} = v_{\pi^{-1}(i)} q_{\sigma^{-1}(j)}$.
    
    \item Allocate item $\pi^{-1}(i)$ to buyer $\pi^{-1}(i)$ for all $i = 1, 2, \ldots, n$.
    
    \item Set price $p_{\pi^{-1}(n)} = v_{n,n}$.
    
    \item Use Lemma~\ref{lemma:nextprice} to find the remaining optimal envy-free prices.
\end{enumerate}

After exhibiting the algorithm designed to solve the problem in question, I will present the main result of this work.

\begin{teorema}
The envy-free perfect matching problem with qualities of items can be solved in $O(n^2)$ time.
\end{teorema}

\begin{proof}
Steps 1 and 2 can be performed in $O(n^2)$ time by any sorting algorithm. Step 3 clearly can be performed in $O(n^2)$ time, since is a construction of a $n$-by-$n$ matrix. Step 4 performs $O(n)$ time operations and Step 5  performs a single $O(1)$ time operation. Finally, Step 6 performs $n-1$ operations of searching maximum values and simple arithmetic computations, so in $O(n^2)$ time. Therefore, the overall time complexity for finding optimal solutions for the envy-free perfect matching problem with qualities of items is $O(n^2)$.
\end{proof}

\section{Conclusions}
\label{sec:conclusions}

In this work, I addressed the envy-free perfect matching problem by incorporating an element into it: the items have quality measures. To take these aspects into account is not an artificial consideration, but it can be found in concrete situations, and even in the literature we can find works that explore these aspects about the items, such as the study by \cite{Chen2016}.

In this way, I showed that the valuation matrix has the inverse Monge property, thus simplifying the search for optimal envy-free allocations and optimal envy-free prices. Furthermore, we showed that this property of the valuation matrix provides an optimal substructure for the problem, enabling the development of a dynamic programming strategy that finds optimal solutions to the problem in quadratic time.

Although many particular cases of the unit-demand envy-free pricing problem can be optimally solved in polynomial-time, an interesting research question is to develop fast heuristics and/or approximation algorithms aiming at experimental treatments of the problem when dealing with instances of large magnitude, such as the proposals by \cite{Shioda2011} and \cite{Myklebust2016}, for example.

This research did not receive any specific grant from funding agencies in the public, commercial, or not-for-profit sectors.

\bibliographystyle{unsrtnat}
\bibliography{references}  






\end{document}